\documentclass[12pt]{article}

\usepackage{amsmath, amsfonts, amsthm}
\usepackage{fullpage}

\newcommand{\ie}{i.\@e.\@}
\newcommand{\eg}{e.\@g.\@}
\newcommand{\etal}{et al.\@}
\newcommand{\lHopital}{l'H\^{o}pital}
\newcommand{\Pade}{Pad\'{e}}
\newcommand{\Plucker}{Pl\"{u}cker}

\newcommand{\realR}{\mathbb{R}}
\newcommand{\compC}{\mathbb{C}}
\newcommand{\intZ}{\mathbb{Z}}
\newcommand{\Prob}{\mathbb{P}}
\newcommand{\WishartW}{\mathcal{W}}
\newcommand{\WishartV}{\mathcal{V}}
\newcommand{\HilH}{\mathbf{H}}
\newcommand{\vect}[1]{\mathbf{#1}}
\DeclareMathOperator{\Tr}{Tr}
\DeclareMathOperator{\vacuum}{vacuum}

\DeclareMathOperator{\diag}{diag}
\DeclareMathOperator{\sgn}{sgn}
\DeclareMathOperator{\Wis}{Wis}
\newtheorem{prop}{Proposition}
\newtheorem{thm}{Theorem}
\newtheorem*{thm*}{Theorem {\ref{thm:first_KP}'}}
\newtheorem*{thmtwo*}{Theorem {\ref{thm:second_KP}'}}
\newtheorem*{thmthree*}{Theorem {\ref{thm:third_KP}'}}

\newtheorem{defn}{Definition}
\newtheorem*{defn*}{Definition {\ref{defn:second_KP}'}}
\newtheorem{lemma}{Lemma}
\newtheorem{remark}{Remark}

\title{Random Matrices with External Source and KP $\tau$ Functions}

\author{Dong Wang \footnote{Centre de recherches math\'{e}matiques, Universit\'{e} de Montr\'{e}al
C. P. 6128, succ. centre ville, Montr\'{e}al, Qu\'{e}bec, Canada H3C 3J7, wangdong@crm.umontreal.ca}}

\begin{document}

\maketitle

\abstract{In this paper we prove that the partition function in the random matrix model with external source is a KP $\tau$ function.}

\section{Introduction}

Let $A \in \HilH^{n \times n}$ be an $n \times n$ Hermitian matrix, and $d\mu(x) = w(x)dx$ be a measure on $\realR$ with all moments finite. Then we define the partition function
\begin{equation} \label{eq:general_definition_of_partition_function}
Z_n(A) = \int_{M \in \HilH^{n \times n}} e^{\Tr(AM)} d\mu(M),
\end{equation}
where $d\mu(M)$ denote a unitary invariant measure on $\HilH^{n \times n}$ such that if eigenvalues of $M$ are $\lambda_1, \dots, \lambda_n$,
\begin{equation}
d\mu(M) = \prod^n_{i=1} w(\lambda_i) dM.
\end{equation}
Due to the unitary-invariance, we assume $A = \diag(a_1, \dots, a_n)$, $a_i \in \realR$ without loss of generality. We consider $Z_n(A)$ as a function of eigenvalues of $A$, and find a KP $\tau$ function property of it.

$Z_n(A)$ arises in the random matrix model with external source \cite{Brezin-Hikami96}, \cite{Brezin-Hikami98}, \cite{Zinn_Justin97}, \cite{Zinn_Justin98}. Let $A \in \HilH^{n \times n}$ be an $n \times n$ Hermitian matrix, and $V(x)$ be a function defined on $\realR$, such that $e^{-V(x)}$ decays sufficiently fast. We consider the ensemble of $n \times n$ Hermitian matrices with the probability density function
\begin{equation} \label{eq:definition_of_density}
P(M) = \frac{1}{Z^V_n(A;\realR)} e^{-\Tr(V(M)-AM)},
\end{equation}
where the normalization constant $Z^V_n(A;\realR)$, also called the partition function, is defined as
\begin{equation} \label{eq:definition_of_partition_function_in_external_source}
Z^V_n(A;\realR) = \int_{M \in \HilH^{n \times n}} e^{-\Tr(V(M)-AM)} dM.
\end{equation}
The ensemble is called the random matrix model with external source, and we can easily identify $Z^V_n(A;\realR)$ as the $Z_n(A)$ in \eqref{eq:general_definition_of_partition_function} with $w(x) = e^{-V(x)}$. In the study of gap probability \cite{Peche06}, \cite{Bleher-Kuijlaars04, Aptekarev-Bleher-Kuijlaars05, Bleher-Kuijlaars07}, \cite{Adler-van_Moerbeke07}, we need to consider also
\begin{equation} \label{eq:definition_of_restricted_partition_function_in_external_source}
Z^V_n(A;E) = \int_{\substack{M \in \HilH^{n \times n} \\ \textnormal{all eigenvalues of $M \in E$}}} e^{-\Tr(V(M)-AM)} dM,
\end{equation}
such that the gap probability that all eigenvalues of $M$ are in $E$, a subset of $\realR$, is $Z^V_n(A;E)/Z^V_n(A;\realR)$. Similarly we identify $Z^V_n(A;E)$ as the $Z_n(A)$ with $w(x) = e^{-V(x)}\chi_{E}(x)$, where $\chi_E$ is the indicator function.

The random matrix model with external source was introduced in the 1990's \cite{Brezin-Hikami96, Brezin-Hikami98} as a generalization of the standard matrix model, \ie\ the $A=0$ case, which is first proposed by Wigner in the 1950's \cite{Mehta04}. However, a special form of the random matrix model with external source has been studied by statisticians since the 1920's, under the name of Wishart ensemble, one of the most important models in multivariate statistics \cite{Muirhead82}, \cite{Goodman63}.

Consider $N$ independent, identically distributed complex samples $\vect{x}_1, \dots, \vect{x}_N$, all of which are $n \times 1$ column vectors, and we further assume that the sample vectors $\vect{x}_i$'s are Gaussian with mean $0$ and covariance matrix $\Sigma$, which is a fixed $n \times n$ positively defined Hermitian matrix. If we put $\vect{x}_i$'s into an $n \times N$ rectangular matrix $X = (\vect{x}_1 : \dots : \vect{x}_N)$, then the sample covariance matrix $S = \frac{1}{N}XX^{\dagger}$ is an $n \times n$ positively defined Hermitian matrix matrix. If we assume $N \geq n$, then the probability density function of $S$ is
\begin{equation}
P(S) = \frac{1}{Z^{\Wis}_{n,N}(\Sigma)} e^{-N\Tr(\Sigma^{-1}S)}(\det S)^{N-n},
\end{equation}
where the normalization constant $Z^{\Wis}_{n,N}(\Sigma)$, analogous to the partition function in \eqref{eq:definition_of_partition_function_in_external_source}, is
\begin{equation}
Z^{\Wis}_{n,N}(\Sigma) = \int_{\substack{S \in \HilH^{n \times n} \\ \textnormal{$S$ is positively defined}}} e^{-N\Tr(\Sigma^{-1}S)}(\det S)^{N-n} dS,
\end{equation}
which is the $Z_n(-N\Sigma^{-1})$ with $w(x) = x^{N-n}\chi_{[0,\infty)}(x)$. To study the distribution of the eigenvalues of $S$, we also need partition functions like the $Z^V_n(A;E)$ in \eqref{eq:definition_of_restricted_partition_function_in_external_source}. See \eg\ \cite{Baik-Ben_Arous-Peche05}, \cite{El_Karoui06}, \cite{Mo08, Mo08a}.

The main result in this paper is that  the partition function $Z_n(A)$ is a KP $\tau$ function. To make the statement precise, we denote power sums of eigenvalues of $A$
\begin{equation} \label{eq:definition_of_p_i}
p_k = \frac{1}{k}\sum^n_{i=1} a^k_i, \qquad k=1,2,\dots.
\end{equation}
Since $Z_n(A)$ is a symmetric function in $a_1, \dots, a_n$, it can be regarded as a function of $\{ p_i \}_{i \in \intZ^+}$. We should be cautious that $p_i$'s are not independent among one another, so that $Z_n(A)$ cannot be written as a function of $p_i$'s in a unique way. However, we will eliminate the ambiguity in the following way. First, we define explicitly a function $\hat{Z}_n(p_1, p_2, \dots)$ in \eqref{eq:definition_of_hat_Z}, and then identify
\begin{equation} \label{eq:relation_betwen_partition_function_and_Schur_expansion}
Z_n(A) = \left. \hat{Z}_n(p_1, p_2, \dots) \right\rvert_{p_k = \frac{1}{k}\sum^n_{i=1} a^k_i}.
\end{equation}
Thus the rigorous statement is
\begin{thm} \label{thm:main_theorem}
$\hat{Z}_n(p_1, p_2, \dots)$ defined by \eqref{eq:definition_of_hat_Z} is a KP $\tau$ function in variables $p_1, p_2, \dots$.
\end{thm}

Since $\hat{Z}_n(p_1, p_2, \dots)$ is a KP $\tau$ function, it satisfies Hirota bilinear equations \cite{Date-Kashiwara-Jimbo-Miwa83}, which are equivalent to a series of PDEs. For example,
\begin{equation}
\left( \frac{\partial^4}{\partial p^4_1} + 3\frac{\partial^2}{\partial p^2_2} - 4\frac{\partial^2}{\partial p_1 \partial p_3}  \right) \log \hat{Z}_n(p_1, p_2, \dots) + 6 \left( \frac{\partial^2}{\partial p^2_1} \log \hat{Z}_n(p_1, p_2, \dots) \right)^2 = 0,
\end{equation}
which corresponds to the first nontrivial Hirota bilinear equation.

In \cite{Zinn_Justin02} Zinn-Justin proved that the HCIZ integral is a Toda $\tau$ function. See also \cite{Zinn_Justin-Zuber03}. This result is closely related to theorem \ref{thm:main_theorem}, since the partition function $Z_n(A)$ is usually evaluated by the HCIZ formula \cite{Brezin-Hikami96}, although we do not use it in this paper, and KP $\tau$ functions are closely related to Toda $\tau$ functions. Theorem \ref{thm:main_theorem} may also be proven by the method in \cite{Zinn_Justin02}.

$Z_n(A)$ was studied as a multi-component KP $\tau$ function by Adler and van Moerbeke \etal\ In \cite{Adler-van_Moerbeke-Vanhaecke09} it is proved that the determinant of a moment matrix for several weights, after adding deformation parameters, is a multi-component KP $\tau$ function. By de Bruijn's formula (lemma \ref{lemma:de_Braijn}), we find that $Z_n(A)$ is a special case of these determinants of moment matrices. Especially, when eigenvalues of $A$ have only two distinct values, $Z_n(A)$ is a $3$-component KP $\tau$ function, and is detailed in \cite{Adler-van_Moerbeke07}. For the nondegenerate case, \ie\ eigenvalues $a_1, \dots, a_n$ of $A$ are distinct, $Z_n(A)$ is a $(1+n)$-component KP $\tau$ function.

The intriguing fact is that our KP $\tau$ function structure is independent of this multi-component KP $\tau$ function structure. In particular, our set of KP flow parameters $\{ p_i \}$ is none of these sets of KP flow parameters in \cite{Adler-van_Moerbeke-Vanhaecke09}.

\begin{remark}
With the help of Virasoro constraints, it is shown in \cite{Adler-van_Moerbeke07} that for quadratic $V$, $Z^V_n(A;E)$ in \eqref{eq:definition_of_restricted_partition_function_in_external_source} satisfies a PDE, which is derived from one Hirota bilinear equation of the multi-component KP $\tau$ function. We can also find Virasoro constraints for $\hat{Z}_n(p_1, p_2, \dots)$ as the KP $\tau$ function and get new PDEs satisfied by $Z^V_n(A;E)$. It will be done in a forthcoming paper.
\end{remark}

In section \ref{Boson-Fermion_correspondence_and_KP_tau_functions} we summarize necessary preliminaries. Then in section \ref{Proof_of_the_main_theorem} we give the definition of $\hat{Z}_n(p_1,p_2,\dots)$ and prove theorem \ref{thm:main_theorem}.

\section{Boson-Fermion correspondence and KP $\tau$ functions} \label{Boson-Fermion_correspondence_and_KP_tau_functions}

The definition of KP $\tau$ functions follows that in \cite{Kac-Raina87}, and all materials on symmetric functions are from \cite{Stanley89}.

KP $\tau$ functions can be defined through representations of the Heisenberg algebra, an infinite dimensional Lie algebra. Over any field $K$ with characteristic $0$, such as $\realR$ or $\compC$, the Heisenberg algebra $H$ denotes the Lie algebra over $K$, generated by $\{ h_k \}_{k \in \intZ}$, satisfying
\begin{equation}
[h_k, h_l] = k \delta_{k, -l}.
\end{equation} 

We can construct a representation of $H$ over the so called Boson Fock space, which is $K[p_1, p_2, \dots]$, the space of polynomials with infinitely many variables. $h_k$'s ($k \leq 0$) act as multiplication operators, and $h_k$'s ($k \geq 1$) act as derivations:
\begin{equation}
h_k \rightarrow 
\begin{cases}
-k p_{-k} & \text{for $k \leq -1$}, \\
1 & \text{for $k = 0$}, \\
\frac{\partial}{\partial p_k} & \text{for $k \geq 1$}.
\end{cases}
\end{equation}
This representation of $H$ is called the Boson representation.

On the other hand, $H$ has another representation over the so called Fermion Fock space. To define the Fermion Fock space, we take an infinite dimensional vector space $V$ with basis $\{ v_i \}_{i \in \intZ}$. The Fermion Fock space $\Lambda$ is composed of semi-infinite forms, spanned by the basis $v_{i_0, i_{-1}, i_{-2}, \dots}$, which are defined as
\begin{equation}
v_{i_0, i_{-1}, i_{-2}, \dots} = v_{i_0} \wedge v_{i_{-1}} \wedge v_{i_{-2}} \wedge \dots,
\end{equation}
with $i_0, i_{-1}, \dots$ strictly decreasing and $i_{-k} = -k$ for $k$ sufficiently large.

We call the semi-infinite form $v_{0, -1, -2, \dots} = v_0 \wedge v_{-1} \wedge v_{-2} \wedge \dots$ the {\em vacuum}, following the physical terminology. Later in this paper, we use {\em form} to mean a semi-infinite form, unless otherwise claimed.

Formally, we define the action of $h_i$ on $V$ by
\begin{equation}
h_k(v_l) = v_{l-k},
\end{equation}
and get the induced action on $\Lambda$ by
\begin{equation}
h_k(v_{i_0, i_{-1}, i_{-2}, \dots}) = \sum_{j \geq 0} v_{i_0} \wedge v_{i_{-1}} \wedge \dots \wedge v_{i_{-j+1}} \wedge v_{i_{-j}-k} \wedge v_{i_{-j-1}} \wedge \dots.
\end{equation}
Although the action of $h_k$ on $V$ is not consistent with the Lie algebra structure of $H$, the action of $h_k$ on $\Lambda$ is a representation of $H$. This is the Fermion representation of $H$.

We can observe that the Boson representation and the Fermion representation of $H$ are equivalent. The correspondence $\Phi$ between $K[p_1, p_2, \dots]$ and $\Lambda$ is
\begin{equation}
\begin{split}
\Phi(1) = & \vacuum, \\
\Phi(p_k/k) = & h_{-k}(\vacuum), \\
\Phi(f(p_1, p_2/2, \dots)) = & f(h_{-1}, h_{-2}, \dots)(\vacuum),
\end{split}
\end{equation}
where $k \geq 0$, and $f$ is a polynomial. Since $[h_k, h_l] = 0$ for $k, l \in \intZ^-$, the polynomial of operators $f(h_{-1}, h_{-2}, \dots)$ is well defined. Although it is not difficult to check the validity of the correspondence $\Phi$, the images of monomials on the Boson Fock space become messy combinations of the basis of $\Lambda$. It is an interesting question what the preimage of $v_{i_0, i_{-1}, i_{-2}, \dots}$ is. The answer is nontrivial, and can be best formulated in notations of symmetric functions.

\begin{prop}
Let $i_0 = \kappa_0$, $i_{-1} = \kappa_1 -1$, $i_{-2} = \kappa_2 -2$, \dots with $\kappa_0 \geq \kappa_1 \geq \kappa_2 \geq \dots$, such that $(\kappa_0, \kappa_1, \dots) = \kappa$ is a partition, then $\Phi^{-1}(v_{i_0, i_{-1}, i_{-2}, \dots}) = \tilde{s}_{\kappa}(p_1, p_2, \dots)$, where $\tilde{s}_{\kappa}$ is a polynomial in the definition of Schur functions by power sums \footnote{Here the term power sum is defined slightly different from the most common definition in symmetric function theory: $p_k = \frac{1}{k}\sum^{\infty}_{i=1} x^k_i$, in the same style as \eqref{eq:definition_of_p_i}.}: If we regard $\{ p_i \}$ as power sums, then the Schur functions $s_{\kappa}$ satisfies
\begin{equation} \label{eq:Schur_polynomial_in_pwer_sums}
s_{\kappa} = \tilde{s}_{\kappa}(p_1, p_2, \dots).
\end{equation}
\end{prop}
Since every symmetric function can be written uniquely as a polynomial of power sums, $\tilde{s}_{\kappa}$ is well defined.

To define KP $\tau$ functions, we need the concept of decomposability of forms. We call $v \in \Lambda$ decomposable, if and only if 
\begin{equation}
v = u_0 \wedge u_{-1} \wedge u_{-2} \wedge \dots,
\end{equation}
where $u_i$'s are linear combinations of $v_i$'s. Now we are ready to give the definition:
\begin{defn} \label{defn:first_KP}
$f(p_1, p_2, \dots) \in K[p_1, p_2, \dots]$ is a KP $\tau$ function if and only if $f = \Phi^{-1}(v)$, where $v$ is a decomposable form. 
\end{defn}

\section{Proof of theorem \ref{thm:main_theorem}} \label{Proof_of_the_main_theorem}

In order to define $\hat{Z}_n(p_1,p_2,\dots)$, we first expand $Z_n(A)$ in Schur polynomials. By the Weyl integral formula, we have
\begin{equation}
Z_n(A) = \frac{1}{C_n} \idotsint \Delta(\lambda)^2 \int_{U(n)} e^{\Tr (AU \diag(\lambda) U^{-1})} dU d\mu(\lambda_1) \dots d\mu(\lambda_n),
\end{equation}
where $C_n$ is a constant, and $dU$ is the Haar measure over $U(n)$. Then we use the identity \cite{Macdonald95}
\begin{equation} \label{eq:the_Schur_polynomial_expanssion}
\int_{U(n)} e^{\Tr (AU \diag(\lambda) U^{-1})} dU = \sum^{\infty}_{k=0} \frac{1}{k!} \sum_{\substack{\kappa \vdash k \\ l(\kappa) \leq n}} \frac{C_{\kappa}(l_1, \dots, l_n) C_{\kappa}(\lambda_1, \dots, \lambda_n)}{C_{\kappa}(1, \dots, 1)},
\end{equation}
where $\kappa$'s are partitions of $k$, $l(\kappa)$ is the length of $\kappa$, and $C_{\kappa}$'s are constant multiples of Schur polynomials $s_{\kappa}$, with the normalization \cite{Dumitriu-Edelman-Shuman07}
\begin{equation}
\sum_{\kappa \vdash k} C_{\kappa}(x_1, \dots, x_n) = (x_1 + \dots + x_n)^k.
\end{equation}
Furthermore, we know \cite{Stanley89}, \cite{Dumitriu-Edelman-Shuman07}:
\begin{equation}
C_{\kappa} = \frac{k!}{H(\kappa)} s_{\kappa} \qquad \textnormal{and} \qquad s_{\kappa}(1, \dots, 1) = \frac{(n)_{\kappa}}{H(\kappa)},
\end{equation}
where $H(\kappa)$ is the hook length product of $\kappa$, and $(n)_{\kappa}$ is the Pochhammer symbol: if $\kappa = (\kappa_1, \kappa_2, \dots, \kappa_l)$, then
\begin{equation}
(n)_{\kappa} = \prod^l_{i=1} \prod^{\kappa_i}_{j=1} (n-i+j).
\end{equation}
Thus if we denote 
\begin{equation}
G_{\kappa} = \idotsint \Delta^2(\lambda) s_{\kappa}(\lambda_1, \dots, \lambda_n) d\mu(\lambda_1) \dots d\mu(\lambda_n),
\end{equation}
we have
\begin{equation} \label{eq:final_formula_of_Wishart}
\begin{split}
Z_n(A)= & \frac{1}{C_n} \idotsint \Delta(\lambda)^2 \sum^{\infty}_{k=0} \sum_{\substack{\kappa \vdash k \\ l(\kappa) \leq n}}  \frac{s_{\kappa}(a_1, \dots, a_n) s_{\kappa}(\lambda_1, \dots, \lambda_n)}{(n)_{\kappa}} d\mu(\lambda_1) \dots d\mu(\lambda_n) \\
= & \frac{1}{C_n} \sum^{\infty}_{k=0} \sum_{\substack{\kappa \vdash k \\ l(\kappa) \leq n}} \frac{1}{(n)_{\kappa}} G_{\kappa} s_{\kappa}(a_1, \dots, a_n),
\end{split}
\end{equation}
and we define
\begin{equation} \label{eq:definition_of_hat_Z}
\hat{Z}_n(p_1,p_2,\dots) =  \frac{1}{C_n} \sum^{\infty}_{k=0} \sum_{\substack{\kappa \vdash k \\ l(\kappa) \leq n}} \frac{1}{(n)_{\kappa}} G_{\kappa} \tilde{s}_{\kappa}(p_1, p_2, \dots),
\end{equation}
where $\tilde{s}_{\kappa}$ is the polynomial defined in \eqref{eq:Schur_polynomial_in_pwer_sums}. Then it is clear that \eqref{eq:relation_betwen_partition_function_and_Schur_expansion} holds.

To determine whether $\hat{Z}_n(p_1,p_2,\dots)$ is a KP function, we turn to the Fermion representation, and map it as 
\begin{equation} \label{eq:Phi_image_of_W}
\Phi(\hat{Z}_n(p_1,p_2,\dots))= \frac{1}{C_n} \WishartV, \qquad \textnormal{with} \qquad \WishartV = \sum^{\infty}_{k=0} \sum_{\substack{\kappa \vdash k \\ l(\kappa) \leq n}} \frac{1}{(n)_{\kappa}} G_{\kappa} v_{\kappa}, 
\end{equation}
where for $\kappa = (\kappa_0, \kappa_1, \dots, \kappa_l)$, we denote $v_{\kappa} = v_{0+\kappa_0, -1+\kappa_1, \dots, -l+\kappa_l, -l-1, \dots}$.

To prove that $\hat{Z}_n(p_1,p_2,\dots)$ is a KP $\tau$ function is equivalent to prove that $\WishartV$ is a decomposable form. We have a simple criterion (\Plucker\ relations) for decomposability of forms, and first introduce two kinds of linear operators $v_i \wedge (v)$ and $\imath_{v^*_i} (v)$ on any $v \in \Lambda$:
\begin{align}
v_i \wedge (v_{i_0} \wedge v_{i_{-1}} \wedge \dots) = & v_i \wedge v_{i_0} \wedge v_{i_{-1}} \wedge \dots, \\
\imath_{v^*_i} (v_{i_0} \wedge v_{i_{-1}} \wedge \dots) = & 
\begin{cases}
(-1)^j v_{i_0} \wedge \dots \wedge v_{i_{-j+1}} \wedge v_{i_{-j-1}} \wedge \dots& \text{if $v_{i_{-j}} = v_i$}, \\
0 & \text{otherwise.} 
\end{cases}
\end{align}
Here we note that $v_i \wedge (v)$ and $\imath_{v^*_i} (v)$ are not in $\Lambda$: they are in $\Lambda^+$ and $\Lambda^-$ respectively. $\Lambda^+$ is spanned by forms $v^+_{i_0, i_{-1}, \dots} = v_{i_0} \wedge v_{i_{-1}} \wedge \dots$ such that $i_0, i_{-1}, \dots$ are strictly decreasing and  $i_{-k} = -k+1$ for $k$ sufficiently large; $\Lambda^-$ is spanned by forms $v^-_{i_0, i_{-1}, \dots} = v_{i_0} \wedge v_{i_{-1}} \wedge \dots$ such that $i_0, i_{-1}, \dots$ are strictly decreasing and $i_{-k} = -k-1$ for $k$ sufficiently large.

Now we can state the criterion \cite{Kac-Raina87}:

\begin{prop}
$v \in \Lambda$ is decomposable, if and only if the tensor of forms in $\Lambda^+ \otimes \Lambda^-$
\begin{equation}
\sum^{\infty}_{i = -\infty} v_i \wedge (v) \otimes \imath_{v^*_i} (v) = 0.
\end{equation}
\end{prop}
For notational simplicity, for a partition $\kappa = (\kappa_0, \dots, \kappa_l)$, we denote $v^+_{\kappa}$ and $v^-_{\kappa}$, analogous to $v_{\kappa}$:
\begin{equation}
v^+_{\kappa} = v_{1+\kappa_0, 0+\kappa_1, \dots, -l+1+\kappa_l, -l, \dots}, \qquad v^-_{\kappa} = v_{-1+\kappa_0, -2+\kappa_1, \dots, -l-1+\kappa_l, -l-2, \dots}.
\end{equation}

For $\kappa$ a partition and $i$ an integer, we define partitions $\kappa+i$ and $\kappa-i$. First, $\kappa = (\kappa_0, \dots, \kappa_l)$ corresponds to a $v^-_{\kappa} \in \Lambda^-$, whose subscript ($-1+\kappa_0, -2+\kappa_1, \dots, -l-1+\kappa_l, -l-2, \dots$) is a decreasing sequence of integers. If $i \in \{ -1+\kappa_0, -2+\kappa_1, \dots, -l-1+\kappa_l, -l-2, \dots \}$, then we say $\kappa+i$ is not well defined; otherwise we can arrange elements in $\{ -1+\kappa_0, -2+\kappa_1, \dots, -l-1+\kappa_l, -l-2, \dots \} \cup \{ i \}$ into a decreasing sequence which is the subscript of a form $v_{\kappa'} \in \Lambda$, and we define $\kappa+i = \kappa'$. Symmetrically, we can define $\kappa-i$: If $i \notin \{ 1+\kappa_0, 0+\kappa_1, \dots, -l+1+\kappa_l, -l, \dots \}$, then $\kappa-i$ is not well defined; otherwise we can arrange $\{ 1+\kappa_0, 0+\kappa_1, \dots, -l+1+\kappa_l, -l, \dots \} \setminus \{ i \}$ into a decreasing sequence, which is the subscript of a $v_{\kappa''} \in \Lambda$, and we define $\kappa-i = \kappa''$.

Now we consider the tensor of forms $\sum^{\infty}_{i = -\infty} v_i \wedge (\WishartV) \otimes \imath_{v^*_i} (\WishartV) \in \Lambda^+ \otimes \Lambda^-$, and compute its coefficients of every $v^+_{\alpha} \otimes v^-_{\beta}$ term. By book-keeping, we get the result
\begin{lemma}
For any two partitions $\alpha$ and $\beta$, the coefficient of the  $v^+_{\alpha} \otimes v^-_{\beta}$ term of $\WishartV$ is
\begin{equation}
c(\alpha, \beta) = \sum_{\substack{i \in \intZ; \\ \textnormal{$\alpha-i$ and $\beta+i$ are both well defined;} \\ l(\alpha-i) \leq n,\ l(\beta+i) \leq n}} \frac{\sgn^+(\alpha, i) \sgn^-(\beta, i)}{(n)_{\alpha-i} (n)_{\beta+i}} G_{\alpha-i} G_{\beta+i}
,
\end{equation}
where $\sgn^+(\alpha, i) = (-1)^j$ if $v^+_{\alpha} = v_{a_0, a_{-1}, \dots}$ and $i = a_{-j}$; $\sgn^-(\beta, i) = (-1)^{j'}$ if $v^-_{\beta} = v_{b_0, b_{-1}, \dots}$ and $b_{-j'} < i < b_{-j'+1}$.
\end{lemma}

To simplify $c(\alpha, \beta)$, we first prove that
\begin{lemma}
For any partitions $\alpha$ and $\beta$, and integer $i$ such that $\alpha-i$ and $\beta+i$ are both well defined, we have for any $n \in \intZ^+$
\begin{equation}
(n)_{\alpha-i}(n)_{\beta+i} = n (n+1)_{\alpha} (n-1)_{\beta}.
\end{equation}
\end{lemma}

\begin{proof}
Let $\alpha = (\alpha_0, \alpha_1, \dots, \alpha_q)$, $\beta = (\beta_0, \beta_1, \dots, \beta_r)$, and ($j = 0, 1, \dots$)
\begin{align}
a_{-j} = &
\begin{cases}
\alpha_j - j + 1 & \textnormal{for $j \leq q$,} \\
-j+1 & \textnormal{for $j > q$;}
\end{cases} \label{eq:definition_of_a_j} \\
b_{-j} = &
\begin{cases}
\beta_j - j -1 & \textnormal{for $j \leq r$,} \\
-j-1 & \textnormal{for $j > r$.}
\end{cases}
\end{align}
We assume that $i = a_{-k}$ and $b_{-l} > i > b_{-l-1}$, then we have
\begin{align}
(n)_{\alpha-i} = & \left( \prod^{k-1}_{j=0} \prod^{n+a_{-j}-1}_{p=n-j} p \right) \left( \prod^{\infty}_{j=k+1} \prod^{n+a_{-j}-1}_{p=n-j+1} p \right), \\
(n)_{\beta+i} = & \left( \prod^{l}_{j=0} \prod^{n+b_{-j}-1}_{p=n-j} p \right) \left( \prod^{n+i-1}_{p=n-l-1} p \right) \left( \prod^{\infty}_{j=l+1} \prod^{n+b_{-j}-1}_{p=n-j-1} p \right).
\end{align}
Here we take the convention that $\prod^{p''}_{p'} = 1$ if $p'' < p'$.

On the other hand, we have
\begin{align}
(n+1)_{\alpha} = & \left( \prod^{k-1}_{j=0} \prod^{n+a_{-j}-1}_{p=n-j+1} p \right) \left( \prod^{n+a_{-k}-1}_{p=n-k-1} p \right) \left( \prod^{\infty}_{j=k+1} \prod^{n+a_{-j}-1}_{p=n-j+1} p \right), \\
(n-1)_{\beta} = & \left( \prod^{l}_{j=0} \prod^{n+b_{-j}-1}_{p=n-j-1} p \right) \left( \prod^{\infty}_{j=l+1} \prod^{n+b_{-j}-1}_{p=n-j-1} p \right).
\end{align}
Thus we get 
\begin{equation}
\frac{(n)_{\alpha-i} (n)_{\beta+i}} {(n+1)_{\alpha} (n-1)_{\beta}} = \frac{\prod^{k-1}_{j=0} (n-j) \prod^{n+i-1}_{n-l-1} p}{\prod^l_{j=0} (n-j-1) \prod^{n+a_{-k}-1}_{n-k+1} p}.
\end{equation}
Notice that $a_{-k} = i$, we can verify that
\begin{equation}
\frac{(n)_{\alpha-i} (n)_{\beta+i}} {(n+1)_{\alpha} (n-1)_{\beta}} = n,
\end{equation}
and prove the lemma.
\end{proof}

From this result we observe that if both the $\alpha-i$ and $\beta+i$ are well defined, the condition $\max(l(\alpha-i), l(\beta+i)) \leq n$ is equivalent to $\max(l(\alpha)-1, l(\beta)+1) \leq n$. Then we have for $l(\alpha) \leq n+1$ and $l(\beta) \leq n-1$,
\begin{equation} \label{eq:simplified_formula_of_c_alpha+beta}
c(\alpha, \beta) = \frac{1}{n (n+1)_{\alpha} (n-1)_{\beta}} \sum_{\substack{i \in \intZ; \\ \textnormal{$\alpha-i$ and $\beta+i$ are} \\ \textnormal{both well defined;} \\ \max(l(\alpha)-1, l(\beta)+1) \leq n}} \sgn^+(\alpha, i) \sgn^-(\beta, i) G_{\alpha-i}G_{\beta+i},
\end{equation}
and later in this section we assume $\max(l(\alpha)-1, l(\beta)+1) \leq n$. We find that for all but finitely many $i \in \intZ$, either $G_{\alpha-i}$ or $G_{\beta+i}$ is not well defined, and we can write \eqref{eq:simplified_formula_of_c_alpha+beta} as a finite summation
\begin{equation} \label{eq:simplified_formula_of_c(alpha,beta)}
c(\alpha, \beta) = \frac{1}{n (n+1)_{\alpha} (n-1)_{\beta}} \sum^n_{j=0} \sgn^+(\alpha, a_{-j}) \sgn^-(\beta, a_{-j}) G_{\alpha-a_{-j}}G_{\beta+a_{-j}},
\end{equation}
where $a_{-j}$ is given by \eqref{eq:definition_of_a_j}, and we assume $G_{\beta+a_{-j}} = 0$ if $\beta+a_{-j}$ is not well defined. 

Now we recall the determinantal formula for Schur polynomials \cite{Stanley99} that if $\kappa = (\kappa_0, \dots, \kappa_l)$ with $l \leq n$, then
\begin{equation}
s_{\kappa}(x_1, \dots, x_n) = \frac{
\begin{vmatrix}
x^{n-1+\kappa_0}_1 & x^{n-1+\kappa_0}_2 & \dots & x^{n-1+\kappa_0}_n \\
x^{n-2+\kappa_1}_1 & x^{n-2+\kappa_1}_2 & \dots & x^{n-2+\kappa_1}_n \\
\vdots & \vdots & \dots & \vdots \\
x^{n-l-1+\kappa_l}_1 & x^{n-l-1+\kappa_l}_2 & \dots & x^{n-l-1+\kappa_l}_n \\
x^{n-l-2}_1 & x^{n-l-2}_2 & \dots & x^{n-l-2}_n \\
\vdots & \vdots & \dots & \vdots \\
1 & 1 & \dots & 1 \\
\end{vmatrix}}{\Delta(x_1, \dots, x_n)}.
\end{equation}
Therefore
\begin{equation} \label{eq:consequence_of_Schur_rep}
\Delta(\lambda)^2 s_{\kappa}(\lambda_1, \dots, \lambda_n) = \Delta(\lambda)
\begin{vmatrix}
\lambda^{n-1+\kappa_0}_1 & \dots & \lambda^{n-1+\kappa_0}_n \\
\lambda^{n-2+\kappa_1}_1 & \dots & \lambda^{n-2+\kappa_1}_n \\
\vdots & \dots & \vdots \\
1 & \dots & 1
\end{vmatrix}.
\end{equation}

To simplify the integrals in \eqref{eq:simplified_formula_of_c_alpha+beta}, we need another formula \cite{de_Bruijn55}
\begin{lemma}[de Bruijn's] \label{lemma:de_Braijn}
For any $f_0, \dots, f_{n-1}, g_0, \dots, g_{n-1} \in L^2(\realR)$, 
\begin{multline}
\idotsint d\lambda_1 \dots d\lambda_n \\
\begin{vmatrix}
f_{n-1}(\lambda_1) & f_{n-1}(\lambda_2) & \dots & f_{n-1}(\lambda_n) \\
f_{n-2}(\lambda_1) & f_{n-2}(\lambda_2) & \dots & f_{n-2}(\lambda_n) \\
\vdots & \vdots & \dots & \vdots \\
f_0(\lambda_1) & f_0(\lambda_2) & \dots & f_0(\lambda_n) \\
\end{vmatrix}
\begin{vmatrix}
g_{n-1}(\lambda_1) & g_{n-1}(\lambda_2) & \dots & g_{n-1}(\lambda_n) \\
g_{n-2}(\lambda_1) & g_{n-2}(\lambda_2) & \dots & g_{n-2}(\lambda_n) \\
\vdots & \vdots & \dots & \vdots \\
g_0(\lambda_1) & g_0(\lambda_2) & \dots & g_0(\lambda_n) \\
\end{vmatrix} \\
= n! \det \left( \int f_i(x)g_j(x) dx \right)_{0 \leq i,j \leq n-1}.
\end{multline}
\end{lemma}
Now we denote
\begin{equation} \label{eq:definition_of_G_i}
G_i = \int x^i d\mu(x),
\end{equation}
and by \eqref{eq:consequence_of_Schur_rep} and lemma \ref{lemma:de_Braijn}, for $0 \leq j \leq n$ we have ($a_{-j}$ is given by \eqref{eq:definition_of_a_j})
\begin{multline} \label{eq:determinantal_formula_of_alpha-i}
G_{\alpha-a_{-j}} = \idotsint \Delta(\lambda)^2 s_{\alpha-a_{-j}}(\lambda_1, \dots, \lambda_n) d\mu(\lambda_1) \dots d\mu(\lambda_n) = \\
n!
\begin{vmatrix}
G_{n-1+a_0} & G_{n-1+a_0+1} & \dots & G_{n-1+a_0+n-1} \\
G_{n-1+a_{-1}} & G_{n-1+a_{-1}+1} & \dots & G_{n-1+a_{-1}+n-1} \\
\vdots & \vdots & \dots & \vdots \\
\hat{G}_{n-1+a_{-j}} & \hat{G}_{n-1+a_{-j}+1} & \dots & \hat{G}_{n-1+a_{-j}+n-1} \\
\vdots & \vdots & \dots & \vdots \\
G_{n-1+a_{-n}} & G_{n-1+a_{-n}+1} & \dots & G_{n-1+a_{-n}+n-1} 
\end{vmatrix},
\end{multline}
where $\hat{\ }$ means the entry is deleted. For $\beta+a_{-j}$, we denote $v^+_{\beta} = b_{b_0, b_{-1}, \dots}$ and similarly have that if $\beta+a_{-j}$ is well defined and $l(\beta+a_{-j}) \leq n$, then
\begin{multline} \label{eq:determinantal_formula_of_beta+i}
\sgn^-(\beta,a_{-j})G_{\beta+a_{-j}} = \sgn^-(\beta,a_{-j}) \idotsint \Delta(\lambda)^2 s_{\beta+a_{-j}}(\lambda_1, \dots, \lambda_n) d\mu(\lambda_1) \dots d\mu(\lambda_n) = \\
n!
\begin{vmatrix}
G_{n-1+a_{-j}} & G_{n-1+a_{-j}+1} & \dots & G_{n-1+a_{-j}+n-1} \\
G_{n-1+b_0} & G_{n-1+b_0+1} & \dots & G_{n-1+b_0+n-1} \\
G_{n-1+b_{-1}} & G_{n-1+b_{-1}+1} & \dots & G_{n-1+b_{-1}+n-1} \\
\vdots & \vdots & \dots & \vdots \\
G_{n-1+b_{-n+2}} & G_{n-1+b_{-n+2}+1} & \dots & G_{n-1+b_{-n+2}+n-1} \\
\end{vmatrix},
\end{multline}
Here we notice that the $n \times n$ matrix in \eqref{eq:determinantal_formula_of_alpha-i} is constructed from an $(n+1) \times n$ matrix with the $(j+1)$-th row eliminated, and the first row in the matrix in \eqref{eq:determinantal_formula_of_beta+i} is the same as the deleted row in the construction of the matrix in \eqref{eq:determinantal_formula_of_alpha-i}. Thus we have a determinantal formula
\begin{multline}
\sum^n_{j=0} \sgn^+(\alpha, a_{-j}) \sgn^-(\beta, a_{-j}) G_{\alpha-a_{-j}} G_{\beta+a_{-j}} = \\
(n!)^2
\begin{vmatrix}
 B_0 & B_1 & \dots & B_{n-1} \\
G_{n-1+b_0} & G_{n-1+b_0+1} & \dots & G_{n-1+b_0+n-1} \\
G_{n-1+b_{-1}} & G_{n-1+b_{-1}+1} & \dots & G_{n-1+b_{-1}+n-1} \\
\vdots & \vdots & \dots & \vdots \\
G_{n-1+b_{-n+2}} & G_{n-1+b_{-n+2}+1} & \dots & G_{n-1+b_{-n+2}+n-1} \\
\end{vmatrix},
\end{multline}
where all rows except for the first one are the same as those in the matrix in \eqref{eq:determinantal_formula_of_beta+i}, and the first row $(B_0, B_1, \dots, B_{n-1})$ is 
\begin{multline} \label{eq:vector_formula_of_B_k}
(B_0, B_1, \dots, B_{n-1}) = \\
\sum^n_{j=0} \sgn^+(\alpha, a_{-j}) G_{\alpha-a_{-j}}(G_{n-1+a_{-j}}, G_{n-1+a_{-j}+1}, \dots, G_{n-1+a_{-j}+n-1}),
\end{multline}
Notice that $\sgn^+(\alpha, a_{-j}) = (-1)^j$, we find that \eqref{eq:vector_formula_of_B_k} is equivalent to that $B_k$ is the determinant of a $(n+1) \times (n+1)$ matrix ($k = 0, 1, \dots, n-1$):
\begin{equation}
B_k = 
\begin{vmatrix}
G_{n-1+a_0+k} & G_{n-1+a_0} & G_{n-1+a_0+1} & \dots & G_{n-1+a_0+n-1} \\
G_{n-1+a_{-1}+k} & G_{n-1+a_{-1}} & G_{n-1+a_{-1}+1} & \dots & G_{n-1+a_{-1}+n-1} \\
\vdots & \vdots & \vdots & \dots & \vdots \\
G_{n-1+a_{-n}+k} & G_{n-1+a_{-n}} & G_{n-1+a_{-n}+1} & \dots & G_{n-1+a_{-n}+n-1} 
\end{vmatrix}.
\end{equation}
Since in the matrix, the first column is identical to the $(k+2)$-th column, we get $B_k = 0$. Therefore $\sum^n_{j=0} \sgn^+(\alpha, a_{-j}) \sgn^-(\beta, a_{-j}) G_{\alpha-a_{-j}} G_{\beta+a_{-j}} = 0$, which means that $c(\alpha,\beta) = 0$ by \eqref{eq:simplified_formula_of_c(alpha,beta)}.

Summing up these results, we get the conclusion that $\WishartV$ is a decomposable form, and by \eqref{eq:Phi_image_of_W} prove theorem \ref{thm:main_theorem}.

\subsection*{Acknowledgements}

I would like to thank M.~Adler, J.~Baik, J.~Harnad, A.~Yu.~Orlov and P.~Zinn-Justin for helpful comments.

\bibliographystyle{plain}
\bibliography{bibliography.bib}
\end{document}